\documentclass[11pt]{article}
\usepackage{color,soul}
\usepackage{titlesec}
\usepackage{hyperref}

\usepackage{mathtools}

\usepackage{graphicx}
\usepackage{amsmath}
\usepackage{amssymb}
\usepackage{amsfonts}
\usepackage{array}
\usepackage{tikz}
\usepackage{hhline}
\usepackage{multirow}
\usepackage{subfig}
\usepackage{hyperref}
\usepackage{mathtools}
\usepackage{tikz}
\usepackage{makecell}
\usepackage{pgfplots}
\usepackage{afterpage}
\usepackage{fancyhdr}
\usepackage{setspace}
\usepackage{algorithmic}
\DeclarePairedDelimiter\ceil{\lceil}{\rceil}
\DeclarePairedDelimiter\floor{\lfloor}{\rfloor}

\newtheorem{proposition} {Proposition}

\newtheorem{theorem}{Theorem}

\newenvironment{proof}{\noindent{\bf Proof:}\indent}%
                      {\hfill $\Box$\par}

\newcommand{\sym}[1]{{\sf #1}}

\title{Separation Results for Boolean Function Classes} 
\author{Aniruddha Biswas and Palash Sarkar \\
Indian Statistical Institute \\
203, B.T.Road, Kolkata \\
India 700108. \\
Email: \{anib\_r, palash\}@isical.ac.in
}
\date{\today}

\begin{document}

\maketitle

\begin{abstract}
	We show (almost) separation between certain important classes of Boolean functions. The technique that we use is to show that the total influence of functions
	in one class is less than the total influence of functions in the other class. In particular, we show (almost) separation of several classes of
	Boolean functions which have been studied in the coding theory and cryptography from classes which have been studied in combinatorics and complexity theory. 
\end{abstract}

\section{Introduction \label{sec-intro} }
Boolean functions are one of the most basic objects in coding theory, cryptography and computer science. Several important classes of Boolean functions have been 
extensively studied in different areas. Some examples in the context of coding theory and cryptography are the classes of bent~\cite{rothaus1976bent} functions, 
plateaued~\cite{zheng1999relationships} functions, functions satisfying the strict avalanche criteria (SAC)~\cite{webster1985design}, and functions satisfying 
propagation characteristics (PC)~\cite{preneel1990propagation}. Examples in the context of combinatorics and complexity theory are monotone functions, 
bounded-depth circuits~\cite{haastad1987computational} and linear threshold functions~\cite{chow1961characterization}.

If the intersection of two classes of Boolean functions is empty, then the classes are disjoint. For two infinite classes of Boolean functions, we say that they are almost 
disjoint, if their intersection is a finite set. Being almost disjoint implies that there is a positive integer $n_0$ such that for any $n\geq n_0$, there is no $n$-variable 
Boolean function which belongs to both the classes. 

Our goal is to show that several pairs of classes of Boolean functions are almost disjoint. We use the notion of total influence~\cite{21923} to show such separation.
The technique that we use is to show that for sufficiently large $n$,
the total influence of any $n$-variable Boolean function in one of the classes is less than the total influence of any $n$-variable Boolean function in the other class.
The specific results that we obtain are the following. The class of Boolean functions consisting of bent functions and functions satisfying SAC and PC is almost
disjoint from the class of monotone functions; the class of functions which can be implemented using constant depth, polynomial size circuits; and the class
of linear threshold functions. Similar separation results are obtained for the class of plateaued functions.

The separation of bent and monotone functions was conjectured in~\cite{celerier2012walsh} and proved in~\cite{carlet2016cryptographic}. Our proof which is based on
total influence is shorter. While~\cite{carlet2016cryptographic} had considered cryptographic properties of monotone functions, the classes of SAC, PC and plateaued
functions were not considered in~\cite{carlet2016cryptographic}. So, the separation results for these classes mentioned above are not present in~\cite{carlet2016cryptographic}.

\section{Background and Notation \label{sec-prelim} }

We use the notation $[n] = \{1, 2,\ldots, n\}$ and $2^{[n]}$ denotes the power set of $[n]$. 
Throughout we write $\log$ for $\log_2$. The expressions $0\log 0$ and $0\log \frac{1}{0}$ are to be interpreted as $0$.

For $\mathbf{x}\in\{-1,1\}^n$, the notation $\mathbf{x}^{\oplus i}$ denotes the vector $(x_1,\ldots , x_{i-1}, -x_i , x_{i+1}, \ldots , x_n)$. 
For $S\subseteq [n]$, the notation $\mathbf{x}^{\oplus S}$ denotes the vector $((-1)^{w_1}x_1,(-1)^{w_2}x_2,\ldots,(-1)^{w_n}x_n)$, where $w_i=1$ if $i\in S$, and $w_i=0$ otherwise. 


We consider the domain of an $n$-variable Boolean function to be $\{-1,1\}^n$ and the co-domain to be $\{-1,1\}$, i.e., an $n$-variable Boolean function $f$ is a map
$f:\{-1,1\}^n\rightarrow \{-1,1\}$. Boolean functions are also considered to be maps from $\{0,1\}^n$ to $\{0,1\}$. The two representations of Boolean functions can be seen to
be equivalent under the map from $\{0,1\}$ to $\{-1,1\}$ which takes $a\in\{0,1\}$ to $(-1)^a$ and the inverse map which takes $b\in\{-1,1\}$ to $(1-b)/2$.

The Fourier transform of $f:\{-1,1\}^n\rightarrow \{-1,1\}$ is a map $\widehat{f}:2^{[n]}\rightarrow [-1,1]$ defined as follows. For $S\subseteq [n]$, 
\begin{eqnarray}\label{eqn-fourier}
	\widehat{f}(S) & = & \frac{1}{2^n} \sum_{\mathbf{x}=(x_1,\ldots,x_n)\in\{-1,1\}^n} f(\mathbf{x}) \prod_{i\in S}x_i.
\end{eqnarray}
From Parseval's identity, we have $\sum_{S \subseteq [n]}\widehat{f}^2(S) = 1$ and so the numbers $\widehat{f}^2(S)$ can be thought of as a probability distribution on $2^{[n]}$.
The Fourier entropy $H(f)$ of $f$ is defined to be the entropy of the probability distribution $\{\widehat{f}^2(S)\}_{S\subseteq [n]}$ and is equal to
\begin{eqnarray}\label{eqn-H}
	H(f) & = & -\sum_{S\subseteq [n]}\widehat{f}^2(S) \log \widehat{f}^2(S).
\end{eqnarray}
The influence of variable $i$ on $f: \{-1,1\}^n \rightarrow \{-1,1\}$ is defined as follows.
\begin{equation*}
    \sym{Inf}_{i}(f) = \Pr_{\mathbf{x} \in \{-1,1\}^n}[f(\mathbf{x}) \neq f(\mathbf{x^{\oplus i}})] = \sum_{\substack{S \subseteq [n] \\ i \in S}}\widehat{f}^2(S).
\end{equation*}
The total influence $I(f)$ of $f: \{-1, 1\}^n \rightarrow \{-1,1\}$ is given by
\begin{equation*}
    I(f) = \sum_{i=1}^n \sym{Inf}_{i}(f) = \sum_{S \subseteq [n]}|S|\widehat{f}^2(S).
\end{equation*}

The influence of a set of variables $S$ on $f: \{-1,1\}^n \rightarrow \{-1,1\}$ is defined to be
\begin{equation*}
	\sym{Inf}_{S}(f) = \Pr_{\mathbf{x} \in \{-1,1\}^n}[f(\mathbf{x}) \neq f(\mathbf{x}^{\oplus S})]
\end{equation*}

In Appendix~\ref{sec-extra}, we have stated the relationships between influence, auto-correlation and average sensitivity of a Boolean function.

\subsection{Some Boolean Function Classes \label{subsec-classes} }
We define the Boolean function classes that will be required in the present work. For some of the classes, the original definitions are different from, though equivalent to, 
the definitions that we provide. 

An $n$-variable Boolean function $f$ is said to be {\em monotone} if the following property holds. For $\mathbf{a},\mathbf{b}\in\{-1,1\}^n$, 
		if $\mathbf{a} \leq \mathbf{b}$ (i.e., $a_i\leq b_i$, $i=1,\ldots,n$), then $f(\mathbf{a})\leq f(\mathbf{b})$. Let $\sym{M}$ denote the set of 
		all monotone Boolean functions.

The following Boolean function classes have been studied in the context of coding theory and cryptography.
\begin{itemize}
   	\item  For even $n$, an $n$-variable Boolean function $f$ is said to be {\em bent}~\cite{rothaus1976bent}, if $\widehat{f}(S) = \pm \frac{1}{2^{n/2}}$, for all $S \in [n]$.
		Let $\sym{B}$ denote the set of all bent functions.
	\item An $n$-variable Boolean function $f$ satisfies the {\em strict avalanche criterion} (SAC)~\cite{webster1985design}, if 
		$\sym{Inf}_{i}(f) = \frac{1}{2}$, for all ${i}$. Further, we say that an $n$-variable function satisfies SAC of order $k$, $0\leq k \leq n-2$, 
		(written as SAC($k$)) 
		if by fixing any $k$ of the $n$ variables to arbitrary values in $\{-1,1\}$, the resulting function satisfies SAC. For $k\geq 0$, let 
		$\sym{S}_k$ denote the set of all Boolean functions satisfying SAC($k$), and define $\sym{S}=\cup_{k\geq 0}\sym{S}_k$.
	\item An $n$-variable Boolean function $f$ satisfies {\em propagation characteristics}~\cite{preneel1990propagation} of degree $k$, $0\leq k \leq n$, 
		(written as PC($k$)) if $\sym{Inf}_{S}(f) = \frac{1}{2}$ for $1 \leq |S| \leq k$. For $k\geq 0$, let $\sym{PC}_k$ denote the set of all Boolean 
		functions satisfying PC$(k)$, and define $\sym{PC}=\cup_{k\geq 0}\sym{PC}_k$. 
	\item An $n$-variable Boolean function $f$ is said to be $k$-{\em plateaued}~\cite{zheng1999relationships}, for $k\in\{0,\ldots,n\}$ and $n\equiv k\bmod 2$, if for 
		all $S \in [n]$, $\widehat{f}^2(S) \in \Big\{0, \pm \frac{1}{2^{(n-k)/2}}\Big\}$. For $k\geq 0$, let $\sym{PL}_k$ be the set of all $k$-plateaued
		Boolean functions, and define $\sym{PL}=\cup_{k\geq 0}\sym{PL}_k$. 
\end{itemize}
We next define some Boolean functions classes which have been studied in the context of complexity theory.
\begin{itemize}
\item 
	A bounded-depth circuit \cite{haastad1987computational} for $n$ variables is a Boolean circuit that consists of AND and OR gates, with inputs $x_1,\ldots,x_n$ and 
	$\overline{x}_1,\ldots,\overline{x}_n$. Fan-in to the 
gates is unbounded but depth is bounded by a constant. Without loss of generality, the circuit is leveled, where gates at level $i$ have all their inputs from level $i-1$; all 
gates at the same level have the same type, i.e., all gates at a particular level are either AND or OR; and the types of gates alternate between AND and OR for successive
levels. The depth of such a circuit is the number of levels that it has. The size of a circuit is the number of gates in it. 
If the size of a bounded depth circuit is bounded by a polynomial in $n$, then it is called an $\sym{AC}^0$ circuit~\cite{linial1993constant}. 
The set of all Boolean functions computable by $\sym{AC}^0$ circuits of depth $d$ is denoted by $\sym{AC}^0[d]$.
\item A Boolean function $f$ is said to be a {\em linear threshold function}~\cite{chow1961characterization}, if there are real constants $w_0,w_1,\ldots,w_n$ such that for 
any $\mathbf{x}=(x_1,\ldots,x_n)\in\{-1,1\}^n$, $f(\mathbf{x})=\sym{sign}(w_0+w_1x_1+\cdots+w_nx_n)$, where $\sym{sign}(z) = 1$ if $z > 0$, and $-1$ if $z \leq 0$.
Let $\sym{LTF}$ denote the set of all linear threshold function.
\end{itemize}
    
%
%
%
Below we collect together some relevant results on total influence that will be required for proving separation results. Some of these results were stated in
terms of average sensitivity which is the same as total influence (see Appendix~\ref{sec-extra}).
\begin{description}
	\item{\textbf{Fact~1}~\cite{boppana1997average}:} For any $n$-variable function $f \in \sym{AC}^0[d]$, $I(f) =  O((\log n)^{d-1})$.
	
	\item{\textbf{Fact~2}~\cite{zhang2011note}:} For any non-constant $n$-variable monotone Boolean function $f$, 
		$I(f)\leq {n \choose \floor{n/2}}\ceil{n/2}/2^{n-1}$. Since ${n \choose \floor{n/2}}\ceil{n/2}/2^{n-1}=\Theta(\sqrt{n})$, we have $I(f)=O(\sqrt{n})$.
	\item{\textbf{Fact~3}~\cite{diakonikolas2014average}:} For any $n$-variable Boolean function $f$ in $\sym{LTF}$, $I(f)\leq 2\sqrt{n}$.
	\item{\textbf{Fact~4}~\cite{gangopadhyay2014fourier}:} For any $n$-variable Boolean function $f$ in $\sym{PL}_k$, $I(f)=\Omega(n-k)$.
\end{description}

\section{Separation Results \label{sec-sep-results} }

For $n\in\mathbb{N}$, let $\mathcal{C}_n^1$ and $\mathcal{C}_n^2$, be two subsets of the set of all $n$-variable Boolean functions. Define 
$\mathcal{C}^1 \triangleq \bigcup_{n \geq 1}\mathcal{C}_n^1$ and $\mathcal{C}^2 \triangleq \bigcup_{n \geq 1}\mathcal{C}_n^2$. Then $\mathcal{C}_1$ and $\mathcal{C}_2$ are 
two infinite classes of Boolean functions. Suppose there exists a constant $n_0$, such that for all ${n \geq n_0}$, $\mathcal{C}_n^1 \bigcap \mathcal{C}_n^2 = \emptyset$. 
If $n_0=1$, then the classes $\mathcal{C}^1$ and $\mathcal{C}^2$ are disjoint. If $n_0>1$, then we say that the classes are $n_0$-disjoint.
Note that if the classes $\mathcal{C}^1$ and $\mathcal{C}^2$ are $n_0$-disjoint, then their intersection is finite, i.e., they are almost disjoint.

Let $\mathcal{C}^1$ and $\mathcal{C}^2$ be two classes of Boolean functions. To show separation between $\mathcal{C}^1$ and $\mathcal{C}^2$ we use the following idea.
Suppose it is possible to find a function $\mathcal{P}$ from the set of all Boolean functions to the reals and a positive integer $n_0$ such that for all 
$n\geq n_0$ and for all $f \in \mathcal{C}_n^1$ and $g \in \mathcal{C}_n^2$, $\mathcal{P}(f) < \mathcal{P}(g)$. Then, it follows that $\mathcal{C}_n^1$ and $\mathcal{C}_n^2$ are 
$n_0$-disjoint. We use total influence as the function $\mathcal{P}$.
To do so, we need results on total influence for both the classes. Results on total influence for
some of the classes have been provided in Section~\ref{sec-prelim}. The following result provides the value of total influence for functions in $\sym{B}\cup \sym{PC}\cup \sym{PL}$.
\begin{proposition}\label{prop_bent}
If an $n$-variable Boolean function $f$ is in $\sym{B}\cup \sym{S} \cup \sym{PC}$, then $I(f)=n/2$.
\end{proposition}
\begin{proof}
	For $f\in \sym{B}$, we have $\widehat{f}(S) = \pm \frac{1}{2^{n/2}}$, for all $S \in [n]$. Hence, 
\begin{equation*}
    I(f) = \sum_{S \subseteq [n]}|S|\widehat{f}(S)^2 = \frac{n}{2}.
\end{equation*}

	Suppose $f$ satisfies SAC($k$) for some $k\geq 0$. Then, it follows that $f$ satisfies SAC($j)$ for $1\leq j<k$ (see~\cite{cusick2017cryptographic}). So, in particular, $f$ satisfies SAC. 
From the definition of SAC, we have $\sym{Inf}_{i}(f) = \frac{1}{2}$, for all ${i}$ and so $I(f) = \frac{n}{2}$.

Suppose $f$ satisfies PC($k)$ for some $k\geq 0$. From the defintion of PC, we have $\sym{Inf}_{S}(f) = \frac{1}{2}$, for $1 \leq |S| \leq k$. 
Therefore, $\sym{Inf}_{\{i\}}(f) = \sym{Inf}_{i}(f) = \frac{1}{2}$, for all $i$ and consequently, $I(f) = \frac{n}{2}$.
\end{proof}

\begin{theorem} \label{thm-1}
	The following disjointness results hold for $\sym{B}\cup \sym{S}\cup \sym{PC}$.
\begin{enumerate}
\item $\sym{M}$ is $4$-disjoint from $\sym{B}\cup \sym{S} \cup \sym{PC}$. 
\item $\sym{LTF}$ is $16$-disjoint from $\sym{B}\cup \sym{S} \cup \sym{PC}$.
\item Let $d$ be any positive integer. Then there exists a positive integer $\mathfrak{n}_0$ (depending on $d$) such that $\sym{AC}^0[d]$ is $\mathfrak{n}_0$-disjoint from 
	$\sym{B}\cup \sym{S} \cup \sym{PC}$.
\end{enumerate}
\end{theorem}
\begin{proof}
\noindent{\em Proof of the first point.}
For any $n$-variable monotone Boolean function $f$, from Fact~2, we have $I(f)\leq {n \choose \floor{n/2}}\ceil{n/2}/2^{n-1}$. Since, for
$n\geq 4$, we have ${n \choose \floor{n/2}}\ceil{n/2}/2^{n-1} < n/2$, it follows that for $n\geq 4$, $I(f)<n/2$. 
On the other hand, from Proposition~\ref{prop_bent} for any $n$-variable Boolean function in $\sym{B}\cup \sym{S} \cup \sym{PC}$, the total influence
is equal to $n/2$. So, $f$ cannot be in $\sym{B}\cup \sym{S} \cup \sym{PC}$. \ \\

\noindent{\em Proof of the second point.}
Let $f$ be any $n$-variable Boolean function in $\sym{LTF}$. From Fact~3, $I(f) \leq 2\sqrt{n}$. Now, for $n > 16$, $2\sqrt{n} < \frac{n}{2}$. Therefore, using 
Proposition~\ref{prop_bent}, we obtain the desired result. \ \\

\noindent{\em Proof of the third point.}
Let $f\in \sym{AC}^0[d]$. From Fact~1, we have $I(f)=O((\log n)^{d-1})$. Consequently, there is a constant $c$ and a positive integer $n_1$, such that
	$I(f)\leq c (\log n)^{d-1}$. Since $d$ is fixed, there is a positive integer $\mathfrak{n}_0$ such that $c (\log n)^{d-1} < n/2$ for all $n\geq \mathfrak{n}_0$. So, for 
	$n\geq \mathfrak{n}_0$, $I(f)<n/2$. From Proposition~\ref{prop_bent}, we have that for $n\geq \mathfrak{n}_0$, $f$ does not belong to $\sym{B}\cup \sym{S} \cup \sym{PC}$.
\end{proof}

The first point of Theorem~\ref{thm-1} provides a shorter proof of the fact that no monotone function on $n\geq 4$ variables is bent, a result which was conjectured 
in~\cite{celerier2012walsh} and originally proved in~\cite{carlet2016cryptographic}. 
The third point of Theorem~\ref{thm-1} shows that in general bent functions and also functions satisfying propagation characteristics and strict avalanche criteria
cannot be realised using constant depth circuits.

The notion of linear threshold function has been extended to polynomial threshold function. An $n$-variable Boolean function $f$ is said to be a degree $d$ polynomial
threshold function (PTF) \cite{bruck1990harmonic} if there is a polynomial $p$ such that $f(\mathbf{x})=\sym{sign}(p(\mathbf{x}))$ for all $\mathbf{x}\in\{-1,1\}$. It has been 
shown in~\cite{diakonikolas2014average} that
if $f$ is a degree $d$ PTF, then $I(f)\leq 2^{O(d)}\cdot \log n\cdot n^{1-1/(4d+2)}$. Let $\sym{PTF}_d$ be the set of all degree $d$ PTFs. 
In a manner similar to the proof of Theorem~\ref{thm-1}, it can be proved that for $d=\log^cn$ with $c<1/2$, the class $\sym{B}\cup \sym{S}\cup \sym{PC}$ is 
almost disjoint from $\sym{PTF}_d$.
It has been conjectured in~\cite{gotsman1994spectral} that if $f$ is any $n$-variable degree $d$ PTF, then $I(f)\leq d\sqrt{n}$. If the conjecture is true, it will show
that $\sym{B}\cup \sym{S} \cup \sym{PC}$ is $4d^2$-disjoint from $\sym{PTF}_d$.

\begin{theorem}\label{thm-2}
Let $k$ be a non-negative integer. The following disjointness results hold for $\sym{PL}_k$.
\begin{enumerate}
\item There exists a positive integer $\mathfrak{n}_0$ such that $\mathcal{M}$ and $\sym{PL}_k$ are $\mathfrak{n}_0$-disjoint.
\item There exists a positive integer $\mathfrak{n}_1$ such that $\sym{LTF}$ and $\sym{PL}_k$ are $\mathfrak{n}_1$-disjoint.
\item Let $d$ be any positive integer. There exists a positive integer $\mathfrak{n}_2$ (depending on $d$) such that $\sym{AC}^0[d]$ and $\sym{PL}_k$ are $\mathfrak{n}_2$-disjoint.
\end{enumerate}
\end{theorem}
\begin{proof}
\noindent{\em Proof of the first point.}
Choose an $\varepsilon \in (1/2,1)$ and let $n$ be a positive integer satisfying $n-n^{\varepsilon}\geq k$. For any $n$-variable function $f$ in $\sym{PL}_k$,
from Fact~4, $I(f) = \Omega(n-k)=\Omega(n^\varepsilon)$. 
So, there is a constant $c_1$ and an integer $n_1$, such that $I(f)\geq c_1 n^{\varepsilon}$ for all $n\geq n_1$. 
Let $g$ be any $n$-variable function in $\sym{M}$.
From Fact~2, we have $I(g)=O(\sqrt{n})$. This implies that there is a constant $c_2$ and an integer $n_2$, such that $I(g)\leq c_2n^{1/2}$ 
for all $n\geq n_2$. Since $\varepsilon>1/2$ and $c_1$ and $c_2$ are constants, there is an integer $\mathfrak{n}_0$ such that $c_1n^{\varepsilon}>c_2n^{1/2}$ for all $n\geq n_0$.
Note that $\mathfrak{n}_0$ has to satisfy $\mathfrak{n}_0-\mathfrak{n}_0^{\varepsilon} \geq k$.
So, for $n\geq \mathfrak{n}_0$, $I(f)\geq c_1 n^{\varepsilon}>c_2n^{1/2}\geq I(g)$ which implies that $f$ cannot be equal to $g$. Consequently, $f$ cannot be in $\sym{M}$. \ \\
	
\noindent{\em Proof of the second point.} The proof is similar to the first point, with the only difference being that Fact~3 is used for the argument
instead of Fact~2. \ \\

\noindent{\em Proof of the third point.} The proof is also similar to the first point, with the difference being that Fact~1 is used for the argument.
\end{proof}

Total influence can be used to separate a few other classes of Boolean functions. We briefly mention these.

An $n$-variable Boolean function is said to have $c$-linearly high entropy~\cite{shalev2018fourier} for real constant $c>0$, if $H(f)\geq cn$. Let 
$c\mbox{-}\sym{LHE}$ denote the set of all Boolean functions having $c$-linearly high entropy. It has been shown in~\cite{shalev2018fourier} that for
$f \in \sym{LHE}_c$, with $c \in (0,\frac{1}{2})$, $H[f] \leq \frac{1+c}{h^{-1}(c^2)}\cdot I(f)$, where $h^{-1}$ is the inverse of binary entropy function.
Consequently, using Fact~1, it follows that for any positive integer $d$ and $c \in (0,\frac{1}{2})$, $\sym{LHE}_c$ and $\sym{AC}^0[d]$ are almost disjoint.

The notion of random LTF was considered in~\cite{chakraborty2018fourier}, where the parameters $w_0,w_1,\ldots,w_n$ are drawn independently from either the uniform
distribution over $[-1,1]$, or from the standard normal distribution. It has been shown~\cite{chakraborty2018fourier} that for an $n$-variable random LTF $f$,
$I(f)=\Omega(\sqrt{n})$ with high probability. Combining with Fact~1 we see that $f$ is not in $\sym{AC}^0[d]$ with high probability, where $d$ is any positive integer.

\section{Conclusion\label{sec-conclu}}
We have used total influence to separate classes of Boolean functions. In particular, we have shown separation of certain classes of Boolean functions of interest
in coding theory and cryptography from classes of Boolean functions which have been considered in combinatorics and complexity theory. 

\bibliographystyle{rp}
\bibliography{draft.bib}

\appendix

\section{Auto-Correlation, Influence and Average Sensitivity \label{sec-extra} }
Given an $n$-variable Boolean function $f$, its auto-correlation function $C_{f}: \{-1, 1\}^n \rightarrow [-2^n, 2^n]$ is defined to be 
\begin{equation*}
	C_{f}(\mathbf{u}) = \sum_{\mathbf{x} \in \{-1,1\}^n}f(\mathbf{x})f(\mathbf{x}\odot \mathbf{u})
\end{equation*}
where $\mathbf{x}\odot \mathbf{u}$ represents the pointwise multiplication of vectors $\mathbf{x}$ and $\mathbf{u}$. 
The auto-correlation function is related to the influence in the following manner. For $S\subseteq [n]$,
\begin{equation*}
    \sym{Inf}_{S}(f) = \frac{1}{2} - \frac{1}{2^{n+1}}C_{f}(\mathbf{w}_S).
\end{equation*}
where $\mathbf{w}_{S}=(w_1,\ldots,w_n) \in \{-1,1\}^n$ is such that $w_i=-1$ if $i\in S$ and $w_i=1$ otherwise.

The sensitivity $s(f,\mathbf{x})$ of a Boolean function $f$ on input $\mathbf{x} = (x_1,x_2,\ldots x_n) \in \{-1, 1\}^n$ is defined in the following manner.
\begin{equation*}
    s(f, \mathbf{x}) = \#\{i \in [n] : f(\mathbf{x}) \neq f(\mathbf{x^{\oplus i}})\}.
\end{equation*}
The average sensitivity $s(f)$ of $f$ is given by
\begin{equation*}
	s(f) =\frac{1}{2^n} \sum_{\mathbf{x} \in \{-1,1\}^n} s(f,\mathbf{x}).
\end{equation*}
It is easy to see that $I(f) = s(f)$~\cite{linial1993constant}.

\end{document}